\documentclass{amsart}

\usepackage{latexsym}
\usepackage{amsmath}
\usepackage{epsfig}
\usepackage{amssymb}
\usepackage{enumerate}
\usepackage[font=scriptsize,labelfont=bf]{caption}
\usepackage{pst-all}
\usepackage{multido}
\usepackage{pst-blur}
\usepackage{color}

\usepackage{amsaddr}

\newtheorem{theorem}{Theorem}[section]
\newtheorem{lemma}[theorem]{Lemma}

\newtheorem{corollary}[theorem]{Corollary}

\newtheorem{example}[theorem]{Example}

\newtheorem*{acknowledgement}{Acknowledgement}

\newrgbcolor{orange}{1 0.3 0.3}
\newrgbcolor{mittelblau}{0.3 0.4 1}
\newrgbcolor{blaugruen}{0 0.5 0.5}
\newrgbcolor{bordeaux}{0.7 0.2 0.4}

\title{Expected Anomalies in the Fossil Record}
\author{Mareike Fischer and Mike Steel}

\address{Allan Wilson Centre for Molecular Ecology and Evolution\\
Biomathematics Research Centre \\University of Canterbury \\
Private Bag 4800 \\ Christchurch, New Zealand}

\date{}

\email{email@mareikefischer.de, m.steel@math.canterbury.ac.nz}



\begin{document}

\begin{abstract}
The problem of intermediates in the fossil record has been
frequently discussed ever since Darwin. The extent of `gaps'
(missing transitional stages) has been used to argue against gradual
evolution from a common ancestor. Traditionally, gaps have often
been explained by the improbability of fossilization and the
discontinuous selection of found fossils. Here we take an analytical
approach and demonstrate why, under certain sampling conditions, we
may not expect intermediates to be found.
 Using a simple null model, we show mathematically that the question of whether a taxon sampled from some
 time in the past is likely to be morphologically intermediate to other samples (dated earlier and later)
 depends on the shape and dimensions of the underlying phylogenetic tree that connects the taxa, and the times
 from which the fossils are sampled. 
\end{abstract}

\maketitle

Corresponding Author:

Mareike Fischer

Phone: +64-3-3667001, Ext. 7695

Fax: +64-3-3642587

Email: email@mareikefischer.de

\newpage

\section{Introduction}

Since Darwin's book {\em On the Origin of Species by Means of
Natural Selection, or the Preservation of Favoured Races in the
Struggle for Life} \cite{darwin}, there has been much debate about
the evidence for continuous evolution from a universal common
ancestor. Initially, Darwin only assumed the relatedness of the
majority of species, not of all of them; later, however, he came to
the view that because of the similarities of all existing species,
there could only be one `root' and one `tree of life' ({\em cf.}
\cite{sober}).   All species are descended from this common ancestor
and indications for their gradual evolution have been sought in the
fossil record ever since. Usually, the improbability of
fossilization or of finding existing fossils was put forward as the
standard answer to the question of why there are so many `gaps' in
the fossil record. Such gaps  have become popularly referred to as
`missing links', i.e. missing intermediates between taxa existing
either today or as fossils.

Of course, the existence of gaps is in some sense inevitable: every
new link gives rise to two new gaps, since evolution is generally a
continuous process whereas fossil discovery will always remain
discontinuous. Moreover, a patchy fossil record is not necessarily
evidence against evolution from a common ancestor through a
continuous series of intermediates -- indeed, in a recent approach,
Elliott Sober ({\em cf.} \cite{sober}) applied simple probabilistic
arguments to conclude that the existence of some intermediates
provides a stronger support for evolution than the non-existence of
any (or some) intermediates could ever provide for a hypothesis of
separate ancestry. Moreover, some lineages appear to be densely sampled, whereas of others only few fossiliferous horizons are known (cf. \cite{schoch}). This problem has been well investigated and statistical models have been developed to master it (see e.g. \cite{marshall_90_1}, \cite{marshall_90_2}), \cite{strauss_sadler}).

In this paper, we suggest a further argument that may help explain
missing links in the fossil record. Suppose that three fossils can
be dated back to three different times. Can we really expect that a
fossil from the intermediate time will appear (morphologically) to
be an `intermediate' of the other two fossils? We will explore this
question via a simple stochastic model.

In order to develop this model, we first state some assumptions we
will make throughout this paper: firstly, we will consider that we
are  sampling fossil taxa of closely related organisms and which differ in a number of morphological
characteristics. We assume this group of taxa has evolved in a
`tree-like' fashion from some common ancestor; that is, there is an
underlying phylogenetic tree, and the taxa are sampled from points
on the branches of this tree.

It is also necessary to say how morphological divergence might be
related to time, as this is important for deciding whether a taxon
is an intermediate or not. In this paper, we make the simplifying
assumption that, within the limited group of taxa under
consideration (and over the limited time period being considered),
the expected degree of morphological divergence between two taxa is
proportional to the total amount of evolutionary history separating
those two taxa. This evolutionary history is simply the time
obtained by adding together the two time periods from the most
recent common ancestor of the two taxa until the times from which
each was sampled (in the case where one taxon is ancestral to the
other, this is simply the time between the two samples). This
assumption on morphological diversity would be valid (in
expectation) if we view morphological distance as being proportional
to the number of discrete characters that two species differ on,
provided that two conditions hold: (i) each character has a constant
rate of character state change (substitution) over the time frame
$T$ that the fossils are sampled from,  and (ii) $T$ is short enough
that the probability of a reverse or convergent change at any given
character is low. We require these conditions to hold in the proofs of the following results. We will discuss other possible relations of
morphological diversification and distance towards the end of this
paper.

 \begin{figure}[ht] \vspace{1cm}
    \begin{minipage}{0.45\textwidth} \hspace{1cm}
            \begin{center}\includegraphics[width=4cm]{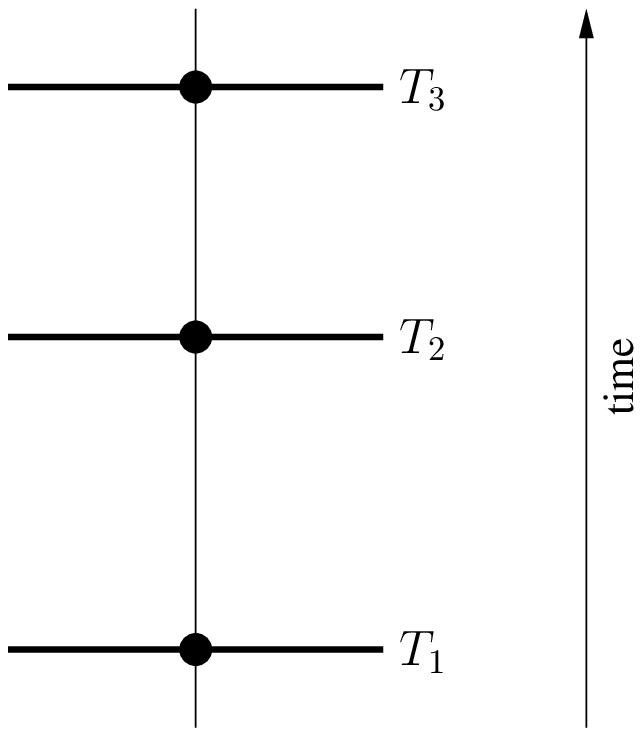}\end{center}
            \renewcommand*{\captionsize}{\footnotesize} \caption{\label{onebranch}
          When the tree consists of only one lineage from which samples are taken at times
          $T_1$, $T_2$ and $T_3$, then clearly the distance $d_{1,3}$ is always larger than $d_{1,2}$ and $d_{2,3}$. Consequently, $E_{1,3}>\max\{E_{1,2}, E_{2,3}\}$. }
    \end{minipage} \hfill
    \begin{minipage}{0.5\textwidth}
      \centering \hfill \vspace{0.3cm}
      \begin{center}\includegraphics[width=6.5cm]{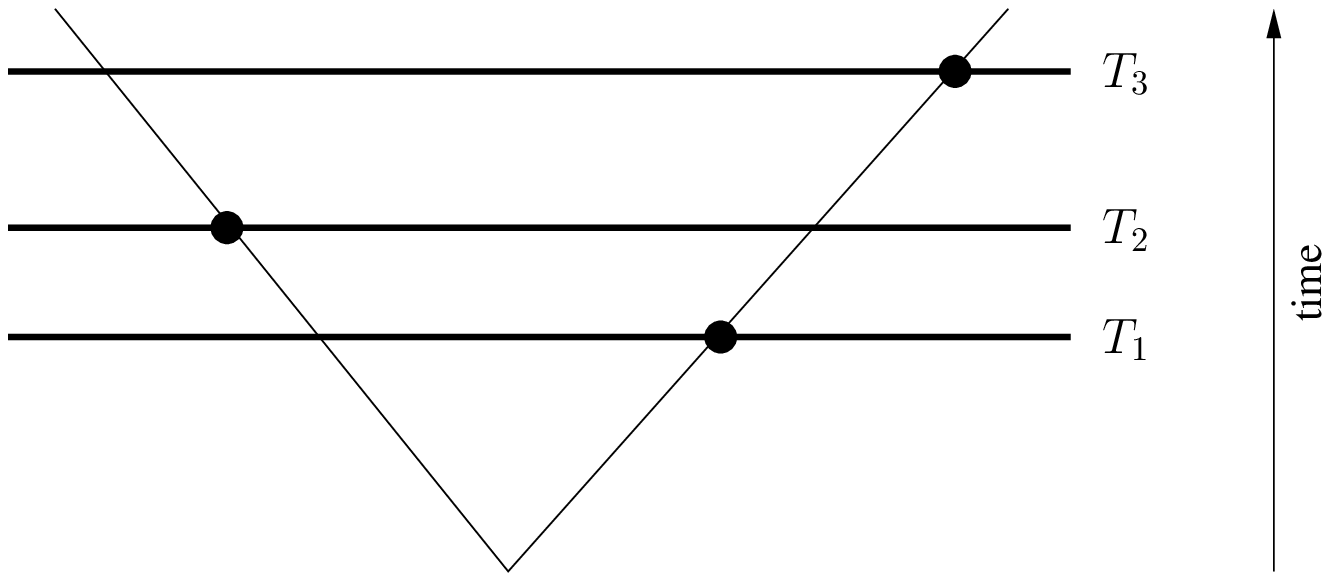}\end{center}
      \renewcommand*{\captionsize}{\footnotesize}
      \caption{ \label{twobranches} For samples taken from different lineages of a tree, the distance $d_{1,3}$ of one particular sample from time $T_1$ to the one of $T_3$ can be smaller than the distance of either of them to the sample taken at time $T_2$.
      Yet in expectation we always have $E_{1,3}>\max \{E_{1,2}, E_{2,3}\}$ for two-branch trees. For more complex trees this can
      fail as we show in Example~\ref{surprise}.
      }
         \end{minipage}
\end{figure}

The simplest scenario is the case where the three samples all lie on
the same lineage, so that the evolutionary tree can be regarded as a
path ({\em cf.} Figure \ref{onebranch}). In this case, the path
distance (and hence expected morphological distance) between the
outer two fossils is always larger than the distance that either of
them has from the fossil sampled from an intermediate time. But for
samples that straddle bifurcations in a tree, it is quite easy to
imagine how this intermediacy could fail; for example, if the two
outer taxa lie on one branch of the tree and the fossil from the
intermediate time lies on another branch far away ({\em cf.} Figure
\ref{twobranches}). But this example might be unlikely to occur, and
indeed we will see that if sampling is uniform across the tree at
any given time, in expectation the morphological distances remain
intermediate even for this case ({\em cf.} Figure
\ref{twobranches}). Yet for more complex trees, this expected
outcome can fail, and perhaps most surprisingly, the distance
between the earliest and latest sample can, in expectation, be  the
{\em smallest} of the three distances in certain extreme cases.

Thus, in order to make general statements, we will consider the
expected degree of relatedness of fossils sampled randomly from
given times. Our results will depend solely on the tree shape
(including branch lengths) of the underlying tree and the chosen
times.

\par

\section{Results}
We begin with some notation. Throughout this paper, we assume a
rooted binary phylogenetic tree to be given with an associated time
scale $0 < T_1 < T_2 < T_3$. The number of $T_i$-lineages (of
lineages extant at time $T_i$) is denoted by $n_i$. For instance, in
Figure \ref{ex1}, the number $n_1$ of $T_1$-lineages is $3$, whereas
the numbers $n_2$ and $n_3$ of $T_2$- and $T_3$-lineages are both
$5$. If not stated otherwise, extinction may occur in the tree.
Every bifurcation in the tree is denoted by $b_i$, where $b_0$ is
the root. Note that in a tree without extinction, the total number
of bifurcations up to time $T_3$ (including the root) is $n_3 - 1$.
For every $b_i$ let $t_i$ denote the time of the occurrence of
bifurcation $b_i$. We may assume that the root is at time $t_0=0$.
\par

Now, for every $b_i$, we make the following definitions:

\begin{equation*} P_{i}^{j,k}:=n_{j,i}^l \cdot n_{k,i}^r + n_{j,i}^r \cdot n_{k,i}^l \hspace{0.5cm} \mbox{for all} \hspace{0.5cm} j,k \in \{1,2,3\}, j \neq k \end{equation*}
where $n_{j,i}^l$ denotes the number of descendants the subtree with root $b_i$ has at time $T_j$ to the left of its root $b_i$, and $n_{j,i}^r$ is defined analogously for the descendants on the right hand side of $b_i$. \\

It can be seen that bifurcations for which at least one branch of
offspring dies out in the same interval where the bifurcation lies
always have $P_{i}^{j,k}$-value $0$. Consequently, if either $t_0 <
t_i < T_1$ or $T_1 < t_i < T_2$ or $T_2 < t_i < T_3$ and one of
$b_i$'s branches becomes extinct in the same interval, respectively,
then $P_{i}^{j,k}$ is $0$ for all $j,k$. Note that the number
$P_{i}^{j,k}$ denotes the number of different paths in the tree from
time $T_j$ to time $T_k$ in the subtree with root $b_i$ and in which
no edge is taken twice.

\par\vspace{0.5cm}
\begin{example}
Consider the tree given in Figure \ref{ex1}. Here, the values
$P_{i}^{j,k}$ for bifurcation $b_1$ corresponding to time $t_1$ are
$P_{1}^{1,2} = n_{1,1}^l \cdot n_{2,1}^r + n_{1,1}^r \cdot n_{2,1}^l
=1 \cdot 2 + 1 \cdot 1 = 3$, $P_{1}^{1,3} = 1 \cdot 3 + 1 \cdot 1 =
4$ and $P_{1}^{2,3} = 1 \cdot 3 + 2 \cdot 1 = 5$.
\par \vspace{0.5cm}
\begin{figure}
\center 
\includegraphics[width=13cm]{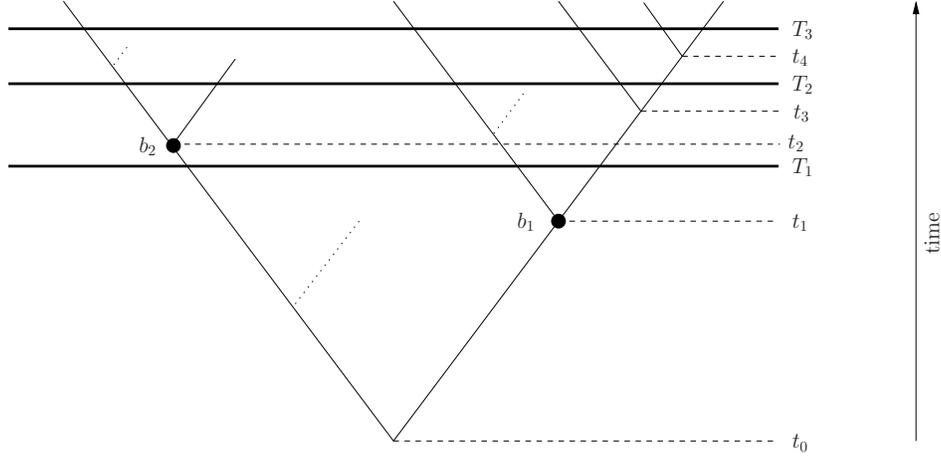}
\caption{\rm{A rooted binary phylogenetic tree with three times $T_1, T_2, T_3$ at
which taxa have been sampled. The dotted branches refer to taxa that
do not contribute to the expected distances from one of these times
to another and thus are not taken into account. On the other hand,
bifurcation $b_2$ at time $t_2$ shows that extinction may have an
impact on the expected values. Such branches have to be
considered.}} \label{ex1}
\end{figure}

\end{example}
\par \vspace{0.5cm}

In the sampling, select uniformly at random one of the
$T_i$-lineages as well as one of the $T_j$-lineages to get the
expected length $E_{i,j}$ of the path  connecting a lineage at time
$T_i$ with  one at time $T_j$ in the underlying phylogenetic tree.
Then, the expectation that a fossil from the intermediate time $T_2$
also will be an intermediate taxon of two taxa taken from $T_1$ and
$T_3$, respectively, refers to the assumption that $E_{1,3} >
\max\{E_{1,2},E_{2,3}\}$. We will show in the following lemma that
this last inequality can fail and describe the precise condition for
this to occur. Moreover, we later show that $E_{1,3}$ can be
strictly smaller (!) than both $E_{1,2}$ and $E_{2,3}$ - that is the
temporally most distant samples can, on average, be more similar
than the temporally intermediate sample is to either of the two. \\
Note that if $P_{i}^{j,k}$ is $0$, the corresponding branch does not
contribute to the expected distance from one time to another. We can
therefore assume without loss of generality that all bifurcations
$b_i$ have at least one descendant on their left-hand side and at
least one on their right-hand side, each in at least one of the
times $T_1,T_2,T_3$. In Figure \ref{ex1}, branches that therefore
need not be considered are represented with dotted lines. \par
\vspace{0.5cm}

In order to simplify the statement of our results, for all bifurcations $b_i$ set \begin{equation*} Q_{i}^{j,k}:=\frac{2\cdot P_{i}^{j,k}}{n_j n_k} \hspace{0.5cm} \mbox{for all} \hspace{0.5cm} j,k \in \{1,2,3\}, j \neq k. \end{equation*}

\begin{lemma} \label{lemma1} Given a rooted binary phylogenetic tree with times $0 < T_1 < T_2 < T_3$ and the root at time $t_0=0$. $\mbox{Then,} \hspace{0.2cm} E_{1,3} \leq E_{1,2} \hspace{0.2cm} \mbox{if and only if}$ \\ \begin{center} $T_3 - T_2 \leq \sum\limits_{i: 0 < t_i < T_1} (Q_i^{1,3} - Q_i^{1,2})t_i $\end{center}
\end{lemma}

\vspace{0.5cm}

\begin{proof} \begin{equation} \label{expected} E_{1,3} = \frac{1}{n_1n_3} \left( \underbrace{n_3 \left( T_3 - T_1\right)}_{\mbox{\tiny every $T_3$-lineage} \atop {\mbox{\tiny has an ancestor} \atop \mbox{\tiny in $T_1$}}}+\hspace{-0.45cm}\sum\limits_{i: 0 < t_i < T_1} \hspace{-0.45cm} \left[P_i^{1,3}\left(T_3-T_1+2\left(T_1-t_i \right) \right)\right] + \underbrace{P_0^{1,3} \left( T_3 + T_1 \right)}_{\mbox{\tiny ways along the root}}\right) \end{equation}

\par\vspace{0.5cm}
In the above bracket, the three summands refer to different paths from time $T_1$ to time $T_3$. The first summand belongs to those paths that go directly from $T_1$ to $T_3$ and thus have length $T_3-T_1$. There are $n_3$ such ways as every $T_3$-lineage  has an ancestor in $T_1$. The second summand sums up all paths going along one of the bifurcations $b_i$ for $i \neq 0$. For every $i$, there are by definition exactly $P_i^{1,3}$ such paths. Similarly, the third summand refers to all paths along the root $b_0$, whose length is determined by taking the distance from $T_1$ to the root plus the distance from there to $T_3$.
\par \vspace{0.5cm}

As there are altogether $n_1n_3$ different paths from $T_1$ to $T_3$
in the tree, we have: \begin{equation} \label{NumberOfWays} n_3 +
\sum\limits_{i: 0< t_i< T_1}P_i^{1,3} + P_0^{1,3} = n_1n_3.
\end{equation} Then, by (\ref{expected}) and (\ref{NumberOfWays}),
we get \begin{equation*}E_{1,3}= \frac{1}{n_1} \cdot \frac{1}{n_3}
\cdot \left(n_1n_3 T_3 + (n_1n_3 - 2 n_3) T_1 - 2 \cdot
\sum\limits_{i:0< t_i < T_1} P_i^{1,3} t_i \right), \end{equation*}
\par and thus \begin{equation} \label{expected13}
 E_{1,3}= T_3 + \frac{n_1-2}{n_1} T_1 - \sum\limits_{i:0< t_i< T_1} Q_i^{1,3}t_i. \end{equation}
Analogously, \begin{equation} \label{expected12} E_{1,2}= T_2 +
\frac{n_1-2}{n_1} T_1 - \sum\limits_{i:0< t_i< T_1} Q_i^{1,2}t_i.
\end{equation} Thus, with (\ref{expected13}) and (\ref{expected12}),
we can conclude: \par \vspace{0.5cm}

$
\begin{array}{ccccc}
E_{1,3} \leq E_{1,2} & \Leftrightarrow & T_3 - \sum\limits_{i:0< t_i< T_1} Q_i^{1,3}t_i & \leq &
T_2 - \sum\limits_{i:0< t_i< T_1} Q_i^{1,2}t_i, \\

&&&\\

& \Leftrightarrow & T_3 - T_2 & \leq & \sum\limits_{i: 0 < t_i < T_1} (Q_i^{1,3} - Q_i^{1,2})t_i.\\
&&&

\end{array}
$

\end{proof}

\par

\begin{corollary} \label{cor1} For a given tree there exist times $0 < T_1 < T_2 < T_3$ such that $E_{1,3} \leq E_{1,2}$ if and only if
$\sum\limits_{i: 0 < t_i < T_1}  (Q_i^{1,3} - Q_i^{1,2})t_i  > 0$. \end{corollary}

\begin{proof}
If $\sum\limits_{i: 0 < t_i < T_1} (Q_i^{1,3} - Q_i^{1,2})t_i  \leq 0$, then by Lemma \ref{lemma1} we need $T_2 \geq T_3$ in order to get $E_{1,3} \leq E_{1,2}$. Hence, there are no values $0 < T_1 < T_2 < T_3$ such that $T_3-T_2$ fulfills the required condition, and so $E_{1,3} > E_{1,2}$ for all choices of $T_i$.
Conversely, suppose  $\sum\limits_{i: 0 < t_i < T_1}  (Q_i^{1,3} - Q_i^{1,2})t_i> 0$. Then, select $T_1, T_2$ with $0<T_1<T_2$ and set \begin{equation*} T_3 := \frac{1}{2} \cdot \sum\limits_{i: 0 < t_i < T_1}(Q_i^{1,3} - Q_i^{1,2}) t_i  + T_2 \end{equation*}
Then, $T_3>T_2$ and  \begin{equation*} T_3-T_2 = \frac{1}{2} \cdot \sum\limits_{i: 0 < t_i < T_1} (Q_i^{1,3} - Q_i^{1,2})t_i
\leq \sum\limits_{i: 0 < t_i < T_1} (Q_i^{1,3} - Q_i^{1,2})t_i .\end{equation*} By Lemma \ref{lemma1}, this choice of $0<T_1<T_2<T_3$ leads to $E_{1,3} \leq E_{1,2}$.

\end{proof}

\begin{corollary} \label{cor2} If either (i) $n_1 = 2$ or (ii) no extinction occurs in the tree and $n_2=n_3$, then $E_{1,3} > E_{1,2}$. \end{corollary}

\begin{proof}
\begin{enumerate}[(i)]
\item Note that if $n_1=2$, obviously only one bifurcation, say $b_{\hat{i}}$ (for some $\hat{i}$ such that $0\leq t_{\hat{i}}<T_1$),
contributes to the number $n_1$ of lineages at time $T_1$, all the
branches added by additional bifurcations become extinct before
$T_1$. Thus: $P_{\hat{i}}^{1,3}, P_{\hat{i}}^{1,2} \neq 0$ and
$P_{i}^{1,3}, P_{i}^{1,2} = 0 \hspace{0.1cm} \mbox{for all}
\hspace{0.1cm} i \neq \hat{i}$.

\noindent Analogously to the proof of Lemma \ref{lemma1} we have for
$n_1=2$: $n_1n_3=2n_3=n_3+P_{\hat{i}}^{1,3}$ and
$n_1n_2=2n_2=n_2+P_{\hat{i}}^{1,2}$. Thus, $n_2=P_{\hat{i}}^{1,2}
\hspace{0.2cm} \mbox{and} \hspace{0.2cm} n_3=P_{\hat{i}}^{1,3}$.
Therefore, $Q_{\hat{i}}^{1,2}=Q_{\hat{i}}^{1,3}=\frac{2}{n_1}$ and
$Q_{i}^{1,2}=Q_{i}^{1,3}=0 \hspace{0.1cm} \mbox{for all}
\hspace{0.1cm} i \neq \hat{i}$. Thus, $\sum\limits_{i: 0 < t_i <
T_1} (Q_i^{1,3} - Q_i^{1,2})t_i =0$ and it follows with Corollary
\ref{cor1} that $E_{1,3} > E_{1,2}$.

\item In this case, obviously $Q_i^{1,2}=Q_i^{1,3} \hspace{0.2cm} \mbox{for all} \hspace{0.2cm} i: 0 < t_i < T_1$ and therefore \\ $\sum\limits_{i: 0<t_i<T_1} (Q_i^{1,3} - Q_i^{1,2})t_i= 0$. Thus, by Corollary \ref{cor1},  $E_{1,3} > E_{1,2}$. \end{enumerate}
\end{proof}

Lemma \ref{lemma1} essentially states that the expected degree of
relatedness from taxa of time $T_1$ to taxa of time $T_3$ can be
larger than the one to taxa of time $T_2$, but it requires the
distance from $T_2$ to $T_3$ to be ``small enough''. Whether such a
solution is feasible can be checked via Corollary \ref{cor1}. Lemma
\ref{lemma1} shows already how the role of intermediates depends on
the times the fossils are taken from. Corollary \ref{cor2}(i) on the
other hand shows how the tree itself has an impact on the expected
values: if the tree shape (including branch lengths) is such that at
time $T_1$ only two taxa exist, then the just mentioned scenario
cannot happen as the condition of Corollary \ref{cor1} is not
fulfilled.

\par However, we can prove an even stronger result,
namely that not only $E_{1,3} < E_{1,2}$ is possible, but $E_{1,3} <
\min\{E_{1,2},E_{2,3}\}$ can be obtained for a suitable choice of
times $T_1, T_2, T_3$. For this, we need the following lemma.
\par\vspace{0.5cm}

\begin{lemma}
\label{lemma2} Given a rooted binary phylogenetic tree with times $0 < T_1 < T_2 < T_3$ and the root at time $t_0=0$. Then  $E_{1,3} \leq E_{2,3}$ if and only if
\begin{equation*}
\frac{n_2-2}{n_2}T_2 - \frac{n_1-2}{n_1} T_1 \hspace{0.2cm} \geq \hspace{0.2cm} \sum\limits_{i: 0 < t_i < T_1} (Q_i^{2,3} - Q_i^{1,3}) t_i +  \sum_{i: T_1 < t_i < T_2} Q_i^{2,3} t_i \end{equation*}

\end{lemma}

\vspace{0.5cm}

\begin{proof}

As in the proof of Lemma \ref{lemma1}, we have ({\em cf.}
(\ref{expected13}))
\begin{equation} E_{1,3}= T_3 + \frac{n_1-2}{n_1} T_1 - \sum\limits_{i:0< t_i< T_1} Q_i^{1,3}t_i. \end{equation}

\begin{equation} \label{expected23} \mbox{Analogously,} \hspace{0.3cm}
E_{2,3}= T_3 + \frac{n_2-2}{n_2} T_2 - \sum\limits_{i:0< t_i< T_2} Q_i^{2,3}t_i. \end{equation}

\begin{equation*} \mbox{Thus,} \hspace{0.3cm}
E_{1,3} \leq E_{2,3}\hspace{0.3cm} \mbox{if and only if} \end{equation*}
\begin{equation*} \frac{n_1-2}{n_1} T_1 - \sum\limits_{i:0< t_i< T_1} Q_i^{1,3}t_i \hspace{0.2cm} \leq \hspace{0.2cm} \frac{n_2-2}{n_2} T_2 - \sum\limits_{i:0< t_i< T_2} Q_i^{2,3}t_i, \end{equation*}
\par\vspace{0.05cm}
\begin{center} which holds precisely if \end{center}

\begin{equation*}
\frac{n_2-2}{n_2}T_2 - \frac{n_1-2}{n_1}T_1 \hspace{0.2cm} \geq \hspace{0.2cm} \sum\limits_{i: 0 < t_i < T_1} (Q_i^{2,3} - Q_i^{1,3}) t_i +  \sum_{i: T_1 < t_i < T_2} Q_i^{2,3} t_i. \end{equation*}

\end{proof}
With the help of the two lemmas we can now state the following
theorem.

\begin{theorem}

Given a rooted binary phylogenetic tree with times $0 < T_1 < T_2 < T_3$ and the root at time $0$. Then, $E_{1,3} \leq \min\{E_{1,2},E_{2,3}\}$ if and only if the following two conditions hold:

\begin{enumerate}[(i)]
\item  \begin{center} $T_3 - T_2 \leq \sum\limits_{i: 0 < t_i < T_1} (Q_i^{1,3} - Q_i^{1,2})t_i $,
\end{center}

\item \begin{center} $ \frac{n_2-2}{n_2}T_2 - \frac{n_1-2}{n_1}T_1 \hspace{0.2cm} \geq \hspace{0.2cm} \sum\limits_{i: 0 < t_i < T_1} (Q_i^{2,3} - Q_i^{1,3}) t_i +  \sum\limits_{i: T_1 < t_i < T_2} Q_i^{2,3} t_i.  $\end{center}

\end{enumerate}

\end{theorem}

\vspace{0.5cm}

\begin{proof} The Theorem follows directly from Lemmas \ref{lemma1} and \ref{lemma2}. \end{proof}

 \par \vspace{0.5cm}
The following example demonstrates the influence of times $0<T_1<T_2<T_3$ according to the above theorem.

\begin{example}
\label{surprise}
Consider again Figure \ref{ex1}.
\begin{enumerate}
\item Assume $t_1=15, T_1=100, t_2=107, t_3=109, T_2=110, T_3=130$. Then, $E_{1,2}=137.33$, $E_{2,3}=155.28$ and $E_{1,3}=155.33$. Hence, for this choice of times, we have $E_{1,3} > \max\{E_{1,2},E_{2,3}\}$.
\item Consider the same times as in the previous case, but choose $T_2=129$ instead of $T_2=110$. This means to move $T_2$ further away from $T_1$ and closer to $T_3$. This change is enough to give completely different expected values: $E_{1,2}=156.33$, $E_{2,3}=166.68$ and $E_{1,3}=155.33$. Hence, for this choice of times, we have $E_{1,3} < \min\{E_{1,2},E_{2,3}\}$.

\end{enumerate}
\end{example}

\par\vspace{0.5cm}

\section{Discussion}
The analysis of the fossil record provides an insight into the
history of species and thus into evolutionary processes. Stochastic models can provide a useful way to infer patters of diversification, and they form a useful link between molecular phylogenetics and paleontology \cite{nee}. Such models would greatly benefit from incorporation of potential fossil ancestors and other extinct data points to infer patterns of evolution. In this paper we have applied a simple model-based phylogenetic approach to study the expected degree of similarity between fossil taxa sampled at intermediate times.

\par `Gaps' in the fossil record are problematic
\cite{schoch} as they can be interpreted as `missing links'.
Therefore, numerous studies concerning the adequacy of the fossil
record have been conducted (see, for example, \cite{durham},
\cite{newell}, \cite{valentine}), and it is frequently found that
even the available fossil record is still incompletely understood.
This is particularly true for ancestor-descendant relationships
(see, for instance,  \cite{engelmann}, \cite{foote}). For example
Foote \cite{foote} reported the probability that a preserved and
recorded species has at least one descendant species that is also
preserved and recorded is on the order of 1\%-10\%. This number is
much higher than the number of identified ancestor-descendant pairs.
Thus, it remains an important challenge to recognize such pairs
\cite{alroy}. This is also essential with regard to
ancestor-intermediate-descendant triplets, as it is possible that
there are in fact fewer `gaps' than currently assumed, i.e. that
intermediates are present but not yet recognized. Such issues have
an important bearing on any conclusions our results might imply
concerning the testing of hypotheses of continuous morphological
evolution, or concerning the shape of the underlying evolutionary
tree based on the non-existence of certain intermediates.
\par

Another challenge is to investigate different phylogenetic models
for describing the expected degree of morphological separation
between different fossil taxa sampled at different times. Our
findings strongly depend on the assumption that morphological
diversification is proportional to the distance in the underlying
phylogenetic tree. This is justified if morphological difference is
proportional to the number of differing discrete characters, that
each of these characters changes at a constant rate over the time
period of sampling, and that homoplasy is rare. This last assumption
requires the rate of character change to be sufficiently small in
relation to the time period of the sampling -- the appearance of
reverse or convergent character states will lead to a more concave
(rather than linear) relationship between morphological divergence
and path distance.  A similar concave relationship might be expected
for continuous morphological evolution as described by neutral
Brownian-motion.

 Thus, the impact of different assumptions on the
role of intermediates could be further investigated. But even if we
assume that diversification is proportional to time, there may be
other ways to measure `distance' that could be usefully explored  --
for instance, one could define the distance between two taxa to be
the maximum (rather than the sum) of the two divergences times of
the taxa back to their most recent common ancestor. This definition
of distance allows the degree of relatedness to be higher for taxa
on the same clade than for other taxa. In this case, there exist analogous results to Lemmas \ref{lemma1} and \ref{lemma2} (results not shown), but the formulae are somewhat different, particularly for Lemma \ref{lemma2}.

\par \vspace{0.3cm}
\begin{acknowledgement}
{\rm We would like to thank Elliott Sober for bringing the
mathematical aspects of intermediates in the fossil record to our
attention, and for helpful comments. We also thank Matt Philips, 
David Penny and two anonymous reviewers for some helpful suggestions.}
\end{acknowledgement}


\end{document}